\documentclass[letterpaper, 10 pt, conference]{ieeeconf}

\IEEEoverridecommandlockouts
\usepackage{graphics}
\usepackage{epsfig}
\usepackage{mathptmx}
\usepackage{times}
\usepackage{amsmath,mathtools}
\usepackage{amssymb}
\usepackage{xcolor}
\usepackage{algorithm}
\usepackage{algpseudocode}
\usepackage[hidelinks]{hyperref}

\newtheorem{theorem}{Theorem}[section]

\newtheorem{lemma}{Lemma}

\newcommand{\Exp}{\operatorname{Exp}}
\newcommand{\Log}{\operatorname{Log}}

\title{\LARGE \bf
Log-linear Backstepping control on $\mathrm{SE}_2(3)$
}

\author{Li-Yu Lin, Benjamin Perseghetti, and James Goppert
\thanks{L. Lin and J. Goppert are with the School of Aeronautics and Astronautics, Purdue University, West Lafayette, IN 47906, USA (e-mail: lin1191@purdue.edu; jgoppert@purdue.edu).}
\thanks{B. Perseghetti is with Rudis Laboratories, Dayton, OH 45342 USA (e-mail: bperseghetti@rudislabs.com).}
\thanks{The authors would like to thank NXP for their support and contribution to the CogniPilot Foundation to help enable this work.}}

\begin{document}

\maketitle

\begin{abstract}
Most of the rigid-body systems which evolve on nonlinear Lie groups where Euclidean control designs lose geometric meaning. In this paper, we introduce a log-linear backstepping control law on $SE_2(3)$ that preserves full rotational–translational coupling. Leveraging a class of mixed-invariant system, which is a group-affine dynamic model, we derive exact logarithmic error dynamics that are linear in the Lie-algebra. The closed-form expressions for the left- and right-Jacobian inverses of $SE_2(3)$ are expressed in the paper, which provides us the exact error dynamics without local approximations. A log-linear backstepping control design ensures the exponentially stable for our error dynamics, since our error dynamics is a block-triangular structure, this allows us to use Linear Matrix Inequality (LMI) formulation or $H_{\inf}$ gain performance design. This work establishes the exact backstepping framework for a class of mixed-invariant system, providing a geometrically consistent foundation for future Unmanned Aerial Vehicle (UAV) and spacecraft control design.
\end{abstract}


\section{Background and Preliminaries}

\subsection{Lie-Group and Invariant Error Concepts}
Rigid-body motion evolves on nonlinear manifolds such as $\mathrm{SO}(3)$ or $\mathrm{SE}(3)$, where Euclidean subtraction loses geometric meaning. A Lie group $G$ provides a globally valid configuration space with associated Lie algebra $\mathfrak g$ (the tangent space at the identity). The exponential and logarithmic maps,
\begin{align}
\Exp: \mathfrak g \!\to\! G, \qquad
\Log: G \!\to\! \mathfrak g,
\end{align}
relate algebra elements to finite motions via $X=\Exp(\xi^\wedge)$.

We will employ the $\mathrm{SE}_2(3)$ Lie group, whose matrix representation is given by:
\begin{equation}
X =
\begin{bmatrix}
R & v & p\\
0 & 1 & 0\\
0 & 0 & 1
\end{bmatrix}
\in \mathrm{SE}_2(3),
\end{equation}
where $p \in \mathbb{R}^3$, inertial-frame velocity $v \in \mathbb{R}^3$, 
and attitude $R \in \mathrm{SO}(3)$,
This group is capable of representing 3D rigid body rotation and translation. In addition, it can handle translational dynamics as we will show in the following. The associated Lie algebra, $\mathfrak{se}_2(3)$, is given by:
\begin{equation}
[x]^\wedge =
\begin{bmatrix}
[\omega]^\wedge & a & v\\
0 & 1 & 0\\
0 & 0 & 1
\end{bmatrix}
\in \mathfrak{se}_2(3),
\end{equation}
where $x = \begin{bmatrix}v & a & \omega\end{bmatrix}^T$ is the element in the Lie algebra, and $[\cdot]^\wedge$ indicates the wedge operator that maps the element from $\mathbb{R}^9$ to the $\mathfrak{se}_2(3)$ Lie algebra. Although the angular velocity, $\omega$, cannot be embedded in the $SE_2(3)$ Lie group, we can consider the dynamics of the angular velocity as a separate sub-system in our application.

For trajectories $X=(p, v, R)^{\wedge}$, $\bar X = (\bar p, \bar v, \bar R)^{\wedge}$ $ \in G$, the \emph{left-invariant} error
\begin{align}
\eta = \bar X^{-1} X,
\end{align}
is unchanged under \emph{left} multiplication of both $X$ and $\bar X$ by the same group element (a change of body frame), whereas the \emph{right-invariant} error $\eta = X \bar X^{-1}$ is invariant under right multiplication.

\vspace{0.3em}
Define the \emph{left-invariant configuration error} between the true and reference states as
\begin{align}
\eta_p &= \bar{R}^\top (p - \bar{p}), \\ \nonumber
\eta_v &= \bar{R}^\top (v - \bar{v}),\\ \nonumber
\eta_R &= \bar{R}^\top R \in \mathrm{SO}(3),
\end{align}
which together form the group element $\eta = [(\eta_p, \eta_v, \eta_R)]^\wedge \in  \mathrm{SE}_2(3)$
satisfying $\eta = I$ when $R=\bar{R}$, $v=\bar{v}$, and $p=\bar{p}$.

For matrices $A,B \in \mathbb{R}^{n\times n}$, the \emph{commutator} (Lie bracket) is
\begin{align}
[A,B] \coloneqq AB - BA,
\end{align}
which satisfies antisymmetry and the Jacobi identity.
The associated \emph{adjoint functions} are
\begin{align}
\operatorname{ad}_A(B)=[A,B], \qquad 
\operatorname{Ad}_{[A]^\vee} B= A^{-1}BA.
\end{align}

The logarithmic (Lie algebra) error is defined as
\begin{equation}
\xi =[ \Log(\eta)]^\vee \in \mathfrak{se}_2(3),
\end{equation}
which provides a minimal coordinate representation of the configuration deviation.

\section{Mixed-Invariant Dynamics and Error Formulation}

Define the augmented group elements
\begin{align}
X =
\begin{bmatrix}
R & v & p\\[3pt]
\mathbf{0}_{1\times 3} & 1 & 0\\
\mathbf{0}_{1\times 3} & 0 & 1
\end{bmatrix},
\qquad
\bar X =
\begin{bmatrix}
\bar R & \bar v & \bar p\\[3pt]
\mathbf{0}_{1\times 3} & 1 & 0\\
\mathbf{0}_{1\times 3} & 0 & 1
\end{bmatrix},
\end{align}
where $X$ is the state of the system and $\bar{X}$ is the reference trajectory.

Consider a specific class of mixed-invariant system:
\begin{equation}
    \dot{X} = (M - C)X + X(N + C),
\end{equation}
where
\begin{align}
M =
\begin{bmatrix}
0 & g & 0
\end{bmatrix}^\wedge,
\qquad \nonumber
N =
\begin{bmatrix}
0 & T e_T & \omega
\end{bmatrix}^\wedge,
\end{align}
\begin{align}
C =
\begin{bmatrix}
\mathbf{0}_{3\times3} & \mathbf{0}_{3\times3} & \mathbf{0}_{3\times1}\\ 
\mathbf{0}_{1\times3} & \mathbf{0}_{1\times3} & 1\\
\mathbf{0}_{1\times3} & \mathbf{0}_{1\times3} & 0
\end{bmatrix}.
\end{align}

where $M, N \in \mathfrak{se}_{2}(3)$ and the system is \emph{mixed-invariant}. The left-invariant term expresses gravity $g$ in the inertial frame, and the right-invariant term expresses thrust input and angular velocity in the body frame, where $e_T$ denotes the unit vector of thrust input.

The left–invariant group error, $\eta$, between the body position and the reference position is defined as
\begin{align}
\eta = \bar{X}^{-1}X \in \mathrm{SE}_{2}(3),
\end{align}
and the algebraic error, $\xi$, as
\begin{align}
\xi = [\Log(\eta)]^\vee \in \mathfrak{se}_2(3).
\end{align}

\begin{lemma}[Logarithm Error dynamics]
Let $G$ be a Lie group with Lie algebra $\mathfrak{g}$ of dimension $n$.
Let $\eta = \bar{X}^{-1}X \in G$ and define the log–error coordinates
$\xi \coloneq [(\Log \eta)]^{\vee} \in \mathbb{R}^n$.
For $\bar N, \bar M, \tilde M,\tilde N \in \mathfrak{g}$ denote their coordinate
vectors by
\begin{align}
\bar n \coloneq [\bar N]^{\vee} = [\mathbf{0}, \bar{T} e_T, \bar{\omega}]^\top,
\bar m \coloneq [\bar M]^{\vee} = [\mathbf{0}, g, 0]^\top, \\ \nonumber
\tilde n \coloneq [\tilde N]^{\vee} = [\mathbf{0}, \tilde{T} e_T, \tilde{\omega}]^\top,
\tilde m \coloneq [\tilde M]^{\vee} = [\mathbf{0}, \tilde g, 0]^\top \in\mathbb{R}^9.
\end{align}
Then the algebraic error $\xi$ satisfies
\begin{align}
\boxed{
\dot{\xi}
= -\operatorname{ad}_{\bar n} \xi
+ ([\xi^\wedge,C])^\vee
+ J_r^{-1}(\xi)\tilde n
+ J_\ell^{-1}(\xi)\operatorname{Ad}_{[{\bar X^{-1}}]^{\vee}}\tilde m  ,
}
\end{align}
where
\begin{align}
J_\ell(\xi) = \frac{\operatorname{ad}_\xi e^{-\operatorname{ad}_\xi}}{I - e^{-\operatorname{ad}_\xi}}, \qquad
J_r(\xi) = \frac{\operatorname{ad}_\xi}{I - e^{-\operatorname{ad}_\xi}},
\end{align}
where $I$ is the $9\times 9$ identity. Here $\operatorname{ad}_{\bar n}$,
$J_\ell(\xi)$, $J_r(\xi)$, and $\operatorname{Ad}_{\bar X^{-1}}^{\vee}$ are
all linear maps $\mathbb{R}^n\!\to\!\mathbb{R}^n$.
\end{lemma}

\begin{proof}
From the definition of $\eta$, we have
\begin{align}
\dot{\eta} = -\bar X^{-1}\dot{\bar X}\bar X^{-1}{X} + \bar X^{-1}\dot{{X}}.
\end{align}
Substituting the expressions for $\dot X$ and $\dot{\bar X}$ gives
\begin{align*}
\dot{\eta}
&= -\bar X^{-1}\!\big((M - C)\bar X + \bar X(N + C)\big)\bar X^{-1}{X} \\begin{align}3pt]
&\quad + \bar X^{-1}\!\big((\bar{M} - C){X} + {X}(\bar{N} + C)\big).
\end{align*}
Let
\begin{align}
\tilde{M} \coloneq \bar{M} - M,
\qquad
\tilde{N} \coloneq \bar{N} - N,
\end{align}
then
\begin{align*}
\dot{\eta}
&= X^{-1}\tilde{M}\bar{X}
+\eta \bar{N} - N\eta
- \eta C + C\eta.
\end{align*}
Left–multiplying by $\eta^{-1}$ yields
\begin{align*}
\eta^{-1}\dot{\eta}
&= \operatorname{Ad}_{\bar X^{-1}}\tilde M + \bar N - \operatorname{Ad}_{\eta^{-1}}N - (I - \operatorname{Ad}_{\eta^{-1}})C .
\end{align*}
Applying the vee operator and the right-trivialized differential of the log map,
\begin{align}
\dot\xi = J_r(\xi)(\eta^{-1}\dot\eta)^{\vee},
\qquad J_r(\xi) = \frac{\operatorname{ad}_\xi}{I - e^{-\operatorname{ad}_\xi}},
\end{align}
with $\operatorname{Ad}_{\eta^{-1}}^{\vee} = e^{-\operatorname{ad}_\xi}$ and 
$J_\ell(\xi) = \operatorname{Ad}_{\eta^{-1}}^{\vee}J_r(\xi)$, gives the stated result.
\end{proof}

The logarithm error dynamics have the explicit form:
\begin{align}
\dot{\xi} &= \begin{bmatrix}
    -[\bar{\omega}]_\times & I & 0 \\
    0 & -[\bar{\omega}]_\times & -[\bar T e_T]_\times \\
    0 & 0 & -[\bar{\omega}]_\times 
\end{bmatrix}\xi \\ \nonumber
&+ J_r^{-1}(\xi)\tilde n
+ J_\ell^{-1}(\xi)\operatorname{Ad}_{[{\bar X^{-1}}]^{\vee}}\tilde m \\ \nonumber
&= (-ad_{\bar n} + A_C)\xi + J_r^{-1}(\xi)\tilde n
+ J_\ell^{-1}(\xi)\operatorname{Ad}_{[{\bar X^{-1}}]^{\vee}}\tilde m
,
\end{align}
where $\bar{n} = [\mathbf{0}, \bar{T}e_T, \bar{\omega}]^\top$ and the
$-\operatorname{ad}_{\bar{n}}\xi$ term generates the frame-rotation
cross-products $-[\bar{\omega}]_\times$ on each component.
The $A_C$ matrix contributes only the kinematic coupling $\dot{\xi}_p = \xi_v$. 
The constant matrix $C$ encodes the kinematic coupling $\dot{\xi}_p = \xi_v$ within the $\mathrm{SE}_{2}(3)$ embedding.  
Although $C\notin\mathfrak{se}_{2}(3)$, it belongs to the normalizer of the algebra,
so $([\xi^\wedge,C])^\vee$ is a \emph{linear} function of $\xi$ that contributes
$A_C\xi = [\xi_v^\top, 0, 0]^\top$.

To express the left/right-Jacobian inverse into matrices form, let $\theta=||\omega||$, $W\equiv\omega^{\times}$, $W^{2}\equiv W~W$. Use series when $\theta$ is small. The $\mathbf{SO(3)}$ left Jacobian is given by:
\begin{align}
    J_{\ell}^{SO3}(\omega)=I+\frac{1-\cos\theta}{\theta^{2}}W+\frac{\theta-\sin\theta}{\theta^{3}}W^{2}
\end{align}
The right Jacobian is related to the left Jacobian by:

\begin{align}
    J_{r}^{SO3}(\omega)=J_{\ell}^{SO3}(-\omega)
\end{align}
The inverse of the left Jacobian is $S_{\ell}(\omega)\triangleq(J_{\ell}^{SO3}(\omega))^{-1}$:

\begin{align}
    S_{\ell}(\omega)=I-\frac{1}{2}W+\left(\frac{1}{\theta^{2}}-\frac{1+\cos\theta}{2\theta\sin\theta}\right)W^{2}
\end{align}
The inverse of the right Jacobian is $S_{r}(\omega)\triangleq(J_{r}^{SO3}(\omega))^{-1}$:

\begin{align}
    S_{r}(\omega)=S_{\ell}(-\omega)
\end{align}

The translation kernels, $Q_r$ and $Q_\ell$, are defined by the integrals:
\begin{align}
    Q_{r}(\omega)&=\int_{0}^{1}s~R(s)ds \\
    Q_{\ell}(\omega)&=\int_{0}^{1}(1-s)~R(s)ds
\end{align}
where $R(s) = \exp(sW)$.
The closed form is given by:
\begin{align}
    Q_{r}(\omega)=q_{0}I+q_{1}W+q_{2}W^{2}
\end{align}
with:
$q_0 = \frac{1}{2}, q_{1}=\frac{\sin\theta-\theta\cos\theta}{\theta^{3}}, q_{2}=\frac{1}{2\theta^{2}} - \frac{\sin\theta}{\theta^3} -\frac{\cos\theta-1}{\theta^4}.$ Then $Q_{\ell}(\omega)$ is given by:
\begin{align}
    Q_{\ell}(\omega)&=J_{\ell}^{SO3}(\omega)-Q_{r}(\omega) \\ \nonumber
    &=\frac{1}{2}I+\left(\frac{1-\cos\theta}{\theta^{2}}-q_{1}\right)W+\left(\frac{\theta-\sin\theta}{\theta^{3}}-q_{2}\right)W^{2}
\end{align}

When we have couplings linear in $p$ or $v$, a handy \begin{align}
\int_{0}^{1}\alpha(s)R(s)x^{\times}R(s)^{\top}ds=\left(\left(\int_{0}^{1}\alpha(s)R(s)ds\right)x\right)^{\times}
\end{align}
This leads to the tensor maps:
\begin{align}
    Q_{\ell}(\omega;x)&=(Q_{\ell}(\omega)x)^{\times} \\
    Q_{r}(\omega;x)&=(Q_{r}(\omega)x)^{\times}
\end{align}

Therefore, the $SE_2(3)$ Jacobian inverse can be written in the matrix form with lower-triangular $9\times9$ matrices in $3\times3$ blocks, ordered $[p, v, \omega]$:
\begin{align}
    J_{\ell}^{-1}(p,v,\omega)&=\begin{bmatrix}
S_{\ell} & -S_{\ell}Q_{\ell}S_{\ell} & -S_{\ell} {Q}_{\ell}(\omega;p)S_{\ell} \\
0 & S_{\ell} & -S_{\ell} {Q}_{\ell}(\omega;v)S_{\ell} \\
0 & 0 & S_{\ell}
\end{bmatrix} \\ \nonumber
 &= \begin{bmatrix}
     c_1 & c_2 & c_3 \\
     0 & c_1 & c_4 \\
     0 & 0 & c_1
 \end{bmatrix}, \\
J_{r}^{-1}(p,v,\omega)&=\begin{bmatrix}
S_{r} & S_{r}Q_{r}S_{r} & -S_{r} {Q}_{r}(\omega;p)S_{r} \\
0 & S_{r} & -S_{r} {Q}_{r}(\omega;v)S_{r} \\
0 & 0 & S_{r}
\end{bmatrix}  \\ \nonumber
&= \begin{bmatrix}
     d_1 & d_2 & d_3 \\
     0 & d_1 & d_4 \\
     0 & 0 & d_1
 \end{bmatrix}
\end{align}
Note that $J_{r}^{-1}(\xi)=J_{\ell}^{-1}(-\xi)$.
Therefore, the logarithm error dynamics can be written as:
\begin{align}
\dot{\xi} &= (-ad_{\bar n} + A_C) \xi + \begin{bmatrix}
    c_3 & c_2 e_T \\
    c_4 & c_1 e_T \\
    c_1 & 0
\end{bmatrix} \begin{bmatrix}
    \tilde \omega \\ \tilde T
\end{bmatrix} + \begin{bmatrix}
    d_2 \\ d_1 \\ 0
\end{bmatrix} \bar R \tilde g
,
\end{align}

\section{Log-Linear Backstepping Control}

The logarithm error dynamics can be seperated as:
\begin{align}
    \dot{\xi_p} &= -[\bar\omega]_\times \xi_p + \xi_v + c_3 \tilde \omega + c_2 e_T \tilde T + d_2 \bar R \tilde g \label{eq:dot_xi_p} \\
    \dot{\xi_v} &= -[\bar\omega]_\times \xi_v + [-\bar T e_T]\xi_r + c_4 \tilde \omega + c_1 e_T \tilde T + d_1 \bar R \tilde g \label{eq:dot_xi_v} \\
    \dot{\xi_r} &= -[\bar\omega]_\times \xi_r + c_1 \tilde \omega \label{eq:dot_xi_r}
\end{align}

To reach the desire position and attitude along the thrust input, we design a feedback control with virtual control points, $\xi^d_r$ and $\xi^d_v$, which represent desired attitude and velocity. We first design $\tilde \omega$ as: 
\begin{align}
    \tilde \omega = c_1^{-1}\left([\bar\omega]_\times \xi_r^d + \dot{\xi_r^d} - K_r(\xi_r - \xi_r^d) \right),
\end{align} 
where $K_r$ is the control gain. Substitute back to \eqref{eq:dot_xi_r}, we then get:
\begin{align}
    \dot{\xi_r} = -[\bar\omega]_\times \xi_r + [\bar\omega]_\times \xi_r^d + \dot{\xi_r^d} - K_r(\xi_r - \xi_r^d),
\end{align}
let $e_r \coloneq \xi_r - \xi_r^d$, the dynamics of the error is:
\begin{align}
    \dot{e_r} = -[\bar\omega]_\times e_r-K_r e_r.
\end{align}

To design the velocity controller, we rewrite \eqref{eq:dot_xi_v} with $e_r$:
\begin{align}
    \dot{\xi_v} &= -[\bar\omega]_\times \xi_v + [-\bar T e_T](\xi_r - \xi_r^d + \xi_r^d) + c_4 \tilde \omega + c_1 e_T \tilde T + d_1 \bar R \tilde g \\ \nonumber
    &= -[\bar\omega]_\times \xi_v + [-\bar T e_T]e_r + [-\bar T e_T]e_r \xi_r^d + c_4 \tilde \omega + c_1 e_T \tilde T + d_1 \bar R \tilde g
\end{align}
The control inputs are $\tilde T$ and $\xi_r^d$, therefore, the control is designed as:
\begin{align}
    [-\bar T e_T]e_r \xi_r^d + c_1 e_T \tilde T = -c_3 \tilde \omega - d_1 \bar R \tilde g - K_v e_v + \dot{\xi_v^d} -[\bar\omega]_\times \xi_v^d
\end{align}
where $\xi_v^d$ represents the desired error velocity between the vehicle and reference, $K_v$ represents the velocity control gain, and $e_v \coloneq \xi_v - \xi_v^d$.
The dynamics of the error dynamics is:
\begin{align}
    \dot{e_v} = -[\bar\omega]_\times e_v - K_v e_v + [-\bar T e_T]_\times e_r
\end{align}

To design the position controller, we rewrite \eqref{eq:dot_xi_p} with the $e_v$:
\begin{align}
\label{eq:dot_xi_p_2}
    \dot{\xi}_p &= -[\bar\omega]_\times \xi_p + \xi_v - \xi_v^d + \xi_v^d + c_3 \tilde \omega + c_2 e_T \tilde T + d_2 \bar R \tilde g \\ \nonumber 
    &= -[\bar\omega]_\times \xi_p + e_v + \xi_v^d + c_3 \tilde \omega + c_2 e_T \tilde T + d_2 \bar R \tilde g 
\end{align}
We design a feedback linearization control:
\begin{align}
    \xi_v^d = -c_3 \tilde \omega - c_2 e_T \tilde T - d_2 \bar R \tilde g - K_p e_p  \label{eq:pos_control}
\end{align}
where $K_p$ represents the position control gain.
Substitute \eqref{eq:pos_control} into \eqref{eq:dot_xi_p_2}:
\begin{align}
    \dot{\xi_p} = -K_p \xi_p + e_v
\end{align}

The entire error dynamics are:
\begin{align}
    \dot{\xi_p} &=  -[\bar\omega]_\times \xi_p -K_p \xi_p + e_v \\
    \dot{e_v} &= -[\bar\omega]_\times e_v - K_v e_v + [-\bar T e_T]_\times e_r \\
    \dot{e_r} &= -[\bar\omega]_\times e_r-K_r e_r.
\end{align}

We could use Linear Matrix Inequalities (LMIs) or $H_{\inf}$ approach to choose gain for the system.

\section{Lyapunov Stability Proof}
\label{sec:stability_inertial}

\begin{theorem}[Exponential stability using inertial derivatives]
Consider the error dynamics:
\begin{align}
    \dot{\xi_p} &=  -[\bar\omega]_\times \xi_p -K_p \xi_p + e_v \\
    \dot{e_v} &= -[\bar\omega]_\times e_v - K_v e_v + [-\bar T e_T]_\times e_r \\
    \dot{e_r} &= -[\bar\omega]_\times e_r-K_r e_r.
\end{align}
where $K_p,K_v,K_r\in\mathbb{R}^{3\times3}$ are symmetric positive-definite gains, 
$\bar\omega(t)\in\mathbb{R}^3$ is the reference angular velocity, and
\begin{align}
B = -\bar T[e_T]_\times,
\end{align}
with $\bar T\ge 0$ and unit vector $e_T$. Assume the state is represented on a single continuous branch of the matrix logarithm. If the gains satisfy
\begin{equation}\label{eq:gain_cond_inertial}
\lambda_{\min}(K_r) \;>\; \frac{\|B\|^2}{2\lambda_{\min}(K_v)},
\end{equation}
then the equilibrium $(\xi_p,e_v,e_r)=(0,0,0)$ is exponentially stable.
\end{theorem}

\begin{proof}
Choose the standard quadratic Lyapunov function
\begin{align}
V(\xi_p,e_v,e_r) = \tfrac12\big(\|\xi_p\|^2 + \|e_v\|^2 + \|e_r\|^2\big).
\end{align}
Compute the time derivative of $V$ along trajectories using the error dynamics:
\begin{align*}
\dot V
&= \xi_p^\top \dot\xi_p + e_v^\top \dot e_v + e_r^\top \dot e_r \\
&= \xi_p^\top\big(-K_p\xi_p + e_v - [\bar\omega]_\times \xi_p\big) \\
&\quad + e_v^\top\big(-K_v e_v + B e_r - [\bar\omega]_\times e_v\big) \\
&\quad + e_r^\top\big(-K_r e_r - [\bar\omega]_\times e_r\big).
\end{align*}

Group terms:
\begin{align}
\dot V &= -\xi_p^\top K_p \xi_p + \xi_p^\top e_v
           - e_v^\top K_v e_v + e_v^\top B e_r - e_r^\top K_r e_r \\
           &\;+\; \underbrace{\big(-\xi_p^\top[\bar\omega]_\times\xi_p - e_v^\top[\bar\omega]_\times e_v - e_r^\top[\bar\omega]_\times e_r\big)}_{=:S}.
\end{align}

Now note the skew-symmetry property: for any vector $x\in\mathbb{R}^3$ and skew matrix $S=[\bar\omega]_\times$,
\begin{align}
x^\top S x = 0.
\end{align}
Therefore each term in $S$ is zero, so $S=0$. Thus
\begin{align}
\dot V = -\xi_p^\top K_p \xi_p + \xi_p^\top e_v
           - e_v^\top K_v e_v + e_v^\top B e_r - e_r^\top K_r e_r.
\end{align}

The minimum eigenvalues $\kappa_p=\lambda_{\min}(K_p)>0$, $\kappa_v=\lambda_{\min}(K_v)>0$, $\kappa_r=\lambda_{\min}(K_r)>0$.
Bound the cross-terms using Young's (Cauchy) inequality:
\begin{align}
\xi_p^\top e_v \le \tfrac{\kappa_p}{2}\|\xi_p\|^2 + \frac{1}{2\kappa_p}\|e_v\|^2,
\end{align}
and
\begin{align}
e_v^\top B e_r \le \tfrac{\kappa_v}{2}\|e_v\|^2 + \frac{\|B\|^2}{2\kappa_v}\|e_r\|^2.
\end{align}
Substituting these bounds yields
\begin{align*}
\dot V &\le -\kappa_p\|\xi_p\|^2 + \left(\tfrac{\kappa_p}{2}\|\xi_p\|^2 + \tfrac{1}{2\kappa_p}\|e_v\|^2\right)\\
            & \quad -\kappa_v\|e_v\|^2 + \left(\tfrac{\kappa_v}{2}\|e_v\|^2 + \tfrac{\|B\|^2}{2\kappa_v}\|e_r\|^2\right) -\kappa_r\|e_r\|^2\\[4pt]
&= -\tfrac{\kappa_p}{2}\|\xi_p\|^2 -\tfrac{\kappa_v}{2}\|e_v\|^2
   -\Big(\kappa_r - \tfrac{\|B\|^2}{2\kappa_v}\Big)\|e_r\|^2.
\end{align*}

By assumption $\kappa_r > \frac{\|B\|^2}{2\kappa_v}$, so define
\begin{align}
\underline\alpha := \min\Big\{\tfrac{\kappa_p}{2},\; \tfrac{\kappa_v}{2},\; \kappa_r - \tfrac{\|B\|^2}{2\kappa_v}\Big\} > 0.
\end{align}
Then $\dot V \le -\underline\alpha(\|\xi_p\|^2 + \|e_v\|^2 + \|e_r\|^2) = -2\underline\alpha V$. Standard comparison yields
\begin{align}
V(t) \le V(0) e^{-2\underline\alpha t},
\end{align}
and therefore there exist constants $c>0,\alpha>0$ such that
\begin{align}
\|[\xi_p(t);\; e_v(t);\; e_r(t)]\| \le ce^{-\alpha t}\|[\xi_p(0);\; e_v(0);\; e_r(0)]\|.
\end{align}
Hence the equilibrium is exponentially stable.
\end{proof}


\end{document}